\documentclass{article}

\title{The Dimension Spectrum Conjecture for Planar Lines}
\author{
D. M. Stull\\
	Department of Computer Science\\
Northwestern University\\
	\texttt{donald.stull@northwestern.edu}
}

\date{}

\usepackage{amsmath}
\usepackage{amssymb}
\usepackage{amsthm} 
\usepackage{caption}

\newtheorem{thm}{Theorem}

\newtheorem{prop}[thm]{Proposition}
\newtheorem{cor}[thm]{Corollary}
\newtheorem{lem}[thm]{Lemma}

\theoremstyle{remark}

\DeclareMathOperator{\Dim}{Dim}

\DeclareMathOperator{\spec}{sp}
\newcommand{\R}{\mathbb{R}}

\newcommand{\N}{\mathbb{N}}
\newcommand{\Q}{\mathbb{Q}}
\newcommand{\ve}{\varepsilon}
\newcommand{\uhr}{{\upharpoonright}}

\begin{document}
	\maketitle

\begin{abstract}
Let $L_{a,b}$ be a line in the Euclidean plane with slope $a$ and intercept $b$. The dimension spectrum $\spec(L_{a,b})$ is the set of all effective dimensions of individual points on $L_{a,b}$. Jack Lutz, in the early 2000s posed the \textit{dimension spectrum conjecture}. This conjecture states that, for every line $L_{a,b}$, the spectrum of $L_{a,b}$ contains a unit interval.

In this paper we prove that the dimension spectrum conjecture is true. Specifically, let $(a,b)$ be a slope-intercept pair, and let $d = \min\{\dim(a,b), 1\}$. For every $s \in [0, 1]$, we construct a point $x$ such that $\dim(x, ax + b) = d + s$. Thus, we show that $\spec(L_{a,b})$ contains the interval $[d, 1+ d]$. 
\end{abstract}

\section{Introduction}
The effective dimension, $\dim(x)$, of a point $x\in \R^n$ gives a fine-grained measure of the algorithmic randomness of $x$. Effective dimension was first defined by J. Lutz \cite{Lutz03a}, and was originally used to quantify the sizes of complexity classes. Unsurprisingly, because of its strong connection to (classical) Hausdorff dimension, effective dimension has proven to be geometrically meaningful \cite{GuLutMayMos14,Mayordomo08,DouLutMauTeut14, LutMay08}. Indeed, an exciting line of research has shown that one can prove classical results in geometric measure theory using effective dimension \cite{LutLut17, Lutz17, LutStu17, LutStu18}. Importantly, these are not effectivizations of known results, but new results whose \textit{proofs} rely on effective methods. Thus, it is of considerable interest to investigate the effective dimensions of points of geometric objects such as lines.

Let $L_{a,b}$ be a line in the Euclidean plane with slope $a$ and intercept $b$. Given the point-wise nature of effective dimension, one can study the \textit{dimension spectrum} of $L_{a,b}$. That is, the set
\begin{center}
$\spec(L_{a,b}) = \{\dim(x,ax+b) \, | x \in \R\}$
\end{center}
of all effective dimensions of points on $L_{a,b}$. In the early 2000s, Jack Lutz posed the \textit{dimension spectrum conjecture} for lines. That is, he conjectured that the dimension spectrum of every line in the plane contains a unit interval.  

The first progress on this conjecture was made by Turetsky. 
\begin{thm}[Turetsky \cite{Turetsky11}]\label{thm:Turetsky}
The set of points $x \in \R^n$ with $\dim(x) = 1$ is connected.
\end{thm}
This immediately implies that $1 \in \spec(L_{a,b})$ for every line $L_{a,b}$. The next progress on the dimension spectrum conjecture was by Lutz and Stull \cite{LutStu17}. They showed that the effective dimension of points on a line is intimately connected to problems in fractal geometry. Among other things, they proved that $1 +d \in \spec(L_{a,b})$ for every line $L_{a,b}$, where $d = \min\{\dim(a,b), 1\}$. Shortly thereafter, Lutz and Stull \cite{LutStu17b} proved the dimension spectrum conjecture for the special case where the effective dimension and strong dimension of $(a,b)$ agree.

In this paper, we prove that dimension spectrum conjecture is true. For every $s \in (0, 1)$, we construct a point $x$ such that $\dim(x, ax + b) = d + s$, where $d = \min\{\dim(a,b), 1\}$. This, combined with the results of Lutz and Stull, imply that 
\begin{center}
$[d, 1 + d] \subseteq \spec(L_{a,b})$, 
\end{center}
for every planar line $L_{a,b}$. The proof of the conjecture builds on the techniques of \cite{LutStu17}. The primary difficulty of the conjecture is the case when the dimension of $x$ is less than the difficulty of the line $(a,b)$. We expand on the nature of this $\dim(x) < \dim(a,b)$ obstacle in Section \ref{ssec:previousworksub}. Our main technical contribution is showing how to overcome this difficulty by encoding the information of $a$ into our point $x$. Further complications arise in the ``high-dimensional" case, i.e., when $\dim(a,b) > 1$. In this case, we combine the encoding idea with a non-constructive argument. 

Apart from its intrinsic interest, recent work has shown that the effective dimensions of points has deep connections to problems in classical analysis \cite{Lutz17, LutStu17, LutStu18, Stull18, LutLut20}. Lutz and Lutz \cite{LutLut17} proved the point-to-set principle, which characterizes the Hausdorff dimension of a \textit{set} by effective dimension of its \textit{individual points}. Lutz and Stull \cite{LutStu17}, using the point-to-set principle, showed that lower bounds on the \textit{effective} dimensions of points on a line are intimately related to well-known problems of classical geometric measure theory such the Kakeya and Furstenberg conjectures.

The structure of the paper is as follows. In Section \ref{sec:preliminaries}, we recall the basic definitions and results of Kolmogorov complexity and effective dimension we need. In Section \ref{sec:previouswork}, we recall the strategy of Lutz and Stull \cite{LutStu17} to give strong lower bounds on the effective dimension of points on a line. In Sections \ref{sec:previouswork} and \ref{ssec:previousworksub} we give intuition about this strategy, and discuss why it is not enough to settle the dimension spectrum conjecture. 

In Section \ref{sec:lowdim}, we prove the dimension spectrum conjecture for lines with effective dimension at most one. We also give a brief overview of this proof, and how it overcomes the strategy discussed in Section \ref{sec:previouswork}. In Section \ref{sec:highdim}, we prove the dimension spectrum conjecture for lines with effective dimension greater than one. We also give intuition of this proof, and how it overcomes the difficulties when the line is high-dimensional. 

Finally, in the conclusion, we discuss open questions and avenues for future research.

\section{Preliminaries}\label{sec:preliminaries}
The \emph{conditional Kolmogorov complexity} of a binary string $\sigma\in\{0,1\}^*$ given binary string $\tau\in\{0,1\}^*$ is 
		\[K(\sigma|\tau)=\min_{\pi\in\{0,1\}^*}\left\{\ell(\pi):U(\pi,\tau)=\sigma\right\}\,,\]
		where $U$ is a fixed universal prefix-free Turing machine and $\ell(\pi)$ is the length of $\pi$. The \emph{Kolmogorov complexity} of $\sigma$ is $K(\sigma)=K(\sigma|\lambda)$, where $\lambda$ is the empty string. Thus, the Kolmogorov complexity of a string $\sigma$ is the minimum length program which, when run on a universal Turing machine, eventually halts and outputs $\sigma$. We stress that the choice of universal machine effects the Kolmogorov complexity by at most an additive constant (which, especially for our purposes, can be safely ignored). See \cite{LiVit08,Nies09,DowHir10} for a more comprehensive overview of Kolmogorov complexity.

We can extend these definitions to Euclidean spaces by introducing ``precision" parameters~\cite{LutMay08,LutLut17}. Let $x\in\R^m$, and $r,s\in\N$. The \emph{Kolmogorov complexity of $x$ at precision $r$} is
\[K_r(x)=\min\left\{K(p)\,:\,p\in B_{2^{-r}}(x)\cap\Q^m\right\}\,.\]
The \emph{conditional Kolmogorov complexity of $x$ at precision $r$ given $q\in\Q^m$} is
\[\hat{K}_r(x|q)=\min\left\{K(p)\,:\,p\in B_{2^{-r}}(x)\cap\Q^m\right\}\,.\]
The \emph{conditional Kolmogorov complexity of $x$ at precision $r$ given $y\in\R^n$ at precision $s$} is
\[K_{r,s}(x|y)=\max\big\{\hat{K}_r(x|q)\,:\,q\in B_{2^{-r}}(y)\cap\Q^n\big\}\,.\]
We abbreviate $K_{r,r}(x|y)$ by $K_r(x|y)$.

		The \emph{effective Hausdorff dimension} and \emph{effective packing dimension}\footnote{Although effective Hausdorff was originally defined by J. Lutz~\cite{Lutz03b} using martingales, it was later shown by Mayordomo~\cite{Mayordomo02} that the definition used here is equivalent. For more details on the history of connections between Hausdorff dimension and Kolmogorov complexity, see~\cite{DowHir10,Mayordomo08}.} of a point $x\in\R^n$ are
	\[\dim(x)=\liminf_{r\to\infty}\frac{K_r(x)}{r}\quad\text{and}\quad\Dim(x) = \limsup_{r\to\infty}\frac{K_r(x)}{r}\,.\]
	Intuitively, these dimensions measure the density of algorithmic information in the point $x$.

	By letting the underlying fixed prefix-free Turing machine $U$ be a universal \emph{oracle} machine, 
	we may \emph{relativize} the definition in this section to an arbitrary oracle set $A \subseteq \N$. The definitions of $K^A_r(x)$, $\dim^A(x)$, $\Dim^A(x)$, etc. are then all identical to their unrelativized versions, except that $U$ is given oracle access to $A$. Note that taking oracles as subsets of the naturals is quite general. We can, and frequently do, encode a point $y$ into an oracle, and consider the complexity of a point \textit{relative} to $y$. In these cases, we typically forgo explicitly referring to this encoding, and write e.g. $K^y_r(x)$.

Among the most used results in algorithmic information theory is the \textit{symmetry of information}. In Euclidean spaces, this was first proved, in a slightly weaker form in \cite{LutLut17}, and in the form presented below in \cite{LutStu17}.
		\begin{lem}\label{lem:unichain}
			For every $m,n\in\N$, $x\in\R^m$, $y\in\R^n$, and $r,s\in\N$ with $r\geq s$,
			\begin{enumerate}
				\item[\textup{(i)}]$\displaystyle |K_r(x|y)+K_r(y)-K_r(x,y)\big|\leq O(\log r)+O(\log\log \|y\|)\,.$
				\item[\textup{(ii)}]$\displaystyle |K_{r,s}(x|x)+K_s(x)-K_r(x)|\leq O(\log r)+O(\log\log\|x\|)\,.$
			\end{enumerate}
		\end{lem}

\subsection{Initial segments versus $K$-optimizing rationals}\label{sec:trunc}
For $x=(x_1,\ldots,x_n)\in\R^n$ and a precision $r\in\N$, let $x\uhr r=(x_1\uhr r,\ldots,x_n\uhr r)$, where each 
\begin{center}
$x_i\uhr r=2^{-r} \lfloor 2^r x_i\rfloor$
\end{center}
is the truncation of $x_i$ to $r$ bits to the right of the binary point. For $r\in(0,\infty)$, let $x\uhr r=x\uhr\lceil r\rceil$. 

We can relate the complexity $K_r(x)$ of $x$ at precision $r$ and the \emph{initial segment complexity} $K(x\uhr r)$ of the binary representation of $x$. Lutz and Stull \cite{LutStu17} proved the following lemma, and its corollaries, relating these two quantities. Informally, it shows that, up to a logarithmic error, the two quantities are equivalent.
	\begin{lem}\label{lem:trunc}
		For every $m,n\in\N$, there is a constant $c$ such that for all $x\in\R^m$, $p\in\Q^n$, and $r\in\N$,
		\[|\hat{K}_r(x|p)-K(x\uhr r\,|\,p)|\leq K(r)+c\,.\]
	\end{lem}
This has the following two useful corollaries.
\begin{cor}\label{cor:trunc}
		For every $m\in\N$, there is a constant $c$ such that for every $x\in\R^m$ and $r\in\N$,
		\[|K_r(x)-K(x\uhr r)|\leq K(r)+c\,.\]
	\end{cor}
	\begin{cor}\label{cor:trunccond}
		For every $m,n\in\N$, there is a constant $c$ such that for all $x\in\R^m$, $y\in\R^n$, and $r,s\in\N$,
		\[|K_{r,s}(x|y)-K(x\uhr r\,|\,y\uhr s)|\leq K(r)+K(s)+c\,.\]
	\end{cor}
	
\section{Previous Work}\label{sec:previouswork}
The proof of our main theorem will use the tools and techniques introduced by Lutz and Stull \cite{LutStu17}. In this section we will state the main lemmas needed for this paper. We will devote some time giving intuition about each lemma. In Subsection \ref{ssec:previousworksub}, we give an informal discussion on how to combine these lemmas to give bounds on the effective dimensions of points on a line. We will also discuss where these tools break down, motivating the techniques introduced in this paper.

The first lemma, informally, states the following. Suppose that $L_{a,b}$ intersects $(x, ax+b)$ and the complexity of $(a,b)$ is low (item (i)). Further assume that (item (ii)), if $L_{u,v}$ is any other line intersecting $(x, ax+b)$ such that $\|(a,b) - (u,v) \| < 2^{-m}$ then either
\begin{enumerate}
    \item $u,v$ is of high complexity, or
    \item $u,v$ is very close to $a,b$.
\end{enumerate}
Then it is possible to compute an approximation of $(a,b)$ given an approximation of $(x, ax+b)$ and first $m$ bits of $(a,b)$. Indeed, we can simply enumerate over all low complexity lines, since we know that the only candidate is essentially $(a,b)$.

\begin{lem}[Lutz and Stull \cite{LutStu17}]\label{lem:point}
				Suppose that $A \subseteq \N$,$a,b,x\in\R$, $m,r\in\N$, $\delta\in\R_+$, and $\ve,\eta\in\Q_+$ satisfy $r\geq \log(2|a|+|x|+5)+1$ and the following conditions.
				\begin{itemize}
					\item[\textup{(i)}]$K^A_r(a,b)\leq \left(\eta+\ve\right)r$.
					\item[\textup{(ii)}] For every $(u,v)\in B_{2^{-m}}(a,b)$ such that $ux+v=ax+b$, \[K^A_{r}(u,v)\geq\left(\eta-\ve\right)r+\delta\cdot(r- t)\,,\]
					whenever $t=-\log\|(a,b)-(u,v)\|\in(0,r]$.
				\end{itemize}
				Then,
				\[K_r^A(a,b, x) \leq K_r(x,ax+b) +  K_{m,r}(a,b \, | \, x,ax+b)+\frac{4\ve}{\delta}r+K(\ve, \eta)+O(\log r)\,.\]
			\end{lem}

The second lemma which will be important in proving our main theorem is the following. It is essentially the approximation version of the simple geometric fact that any two lines intersect at a single point. In other words, if $ax + b = ux + v$ and you are given an approximation of $(a,b)$ and an approximation of $(u,v)$, then you can compute an approximation of $x$. Moreover, the quality of the approximation of $x$ depends linearly on the distance between $(u,v)$ and $(a,b)$.
\begin{lem}[\cite{LutStu17}]\label{lem:lines}
			Let $a,b,x\in\R$. For all $u,v\in B_1(a,b)$ such that $u x+v=ax+b$, and for all $r\geq t:=-\log\|(a,b)-(u,v)\|$,
			\[K_{r}(u,v)\geq K_t(a,b) + K_{r-t,r}(x|a,b)-O(\log r)\,.\]
		\end{lem}
The primary function of this lemma is to give a lower bound on the complexity of any line intersecting $(x, ax+b)$, i.e., ensuring condition (ii) of the previous lemma.

Finally, we also need the following oracle construction of Lutz and Stull. The purpose of this lemma is to show that we can lower the complexity of our line $(a,b)$, thus ensuring item (i) of Lemma \ref{lem:point}. Crucially, we can lower this complexity using only the information contained in $(a,b)$.
\begin{lem}[\cite{LutStu17}]\label{lem:oracles}
				Let $r\in\N$, $z\in\R^2$, and $\eta\in\Q\cap[0,\dim(z)]$.
				Then there is an oracle $D=D(r,z,\eta)$ satisfying
				\begin{itemize}
					\item[\textup{(i)}] For every $t\leq r$, $K^D_t(z)=\min\{\eta r,K_t(z)\}+O(\log r)$.
					\item[\textup{(ii)}] For every $m,t\in\N$ and $y\in\R^m$,
					$K^{D}_{t,r}(y|z)=K_{t,r}(y|z)+O(\log r)$
					and
					$K_t^{z,D}(y)=K_t^z(y)+O(\log r)$.
				\end{itemize}
			\end{lem}
		
\subsection{Combining the lemmas}\label{ssec:previousworksub}
We now briefly discuss the strategy of \cite{LutStu17} which combines the above lemmas to give non-trivial bounds on the effective dimension of points on a line. Suppose $(a,b)$ is a line with $\dim(a,b) = d$, and  $x$ is a point with $\dim^{a,b}(x) = s$. We will also make the crucial assumption that $d\leq s$. Roughly, Lutz and Stull showed that, for sufficiently large $r$
\begin{center}
    $K_r(x, ax + b) \geq (s + d) r$.
\end{center}
The strategy is as follows. Note that to simplify the exposition, all inequalities in this discussion will be approximate. Using Lemma \ref{lem:oracles}, we find an oracle $D$ which reduces the complexity of $(a,b)$ to some $\eta \leq dr$, i.e., $K^D_r(a,b) = \eta r$. Combining this with Lemma \ref{lem:lines}, we get a lower bound on every line $(u,v)$ intersecting $(x, ax+b)$. That is, we show for any such line,
\[K^D_{r}(u,v)\geq \eta t + s(r - t)-O(\log r)\]
By our choice of $\eta$, we can simplify this inequality to get 
\[K^D_{r}(u,v)\geq sr -O(\log r)\]
In particular, relative to $D$, both conditions of Lemma \ref{lem:point} are satisfied and we have the sufficient lower bound.

In the previous sketch, it was crucial that the dimension of $(a,b)$ was less than $s$, in order for the lower bound from Lemma \ref{lem:lines} to be useful. In the case where $\dim(a,b)$ is much larger than $\dim(x)$, this strategy breaks down, and further techniques are required. 

We also note that this seems to be a very deep issue. As discussed in the Introduction, the point-to-set principle of J. Lutz and N. Lutz \cite{LutLut17} allows us to translate problems from (classical) geometric measure theory into problems of effective dimension. The same issue discussed in this section occurs when attacking the notorious Kakeya and Furstenberg set conjectures using the point-to-set principle. While resolving this obstacle in full generality is still elusive, we are able to get around it in the context of the Dimension Spectrum Conjecture.

\section{Low-Dimensional Lines}\label{sec:lowdim}
In this section, we prove the spectrum conjecture for lines with $\dim(a,b) \leq 1$.
\begin{thm}\label{thm:maintheoremLow}
Let $(a,b) \in \R^2$ be a slope-intercept pair with $\dim(a,b)\leq 1$.. Then for every $s \in [0, 1]$, there is a point $x\in \R$ such that
\begin{center}
$\dim(x, ax + b) = s + \dim(a,b)$.
\end{center}
\end{thm}	

We begin by giving an intuitive overview of the proof.
\subsection{Overview of the proof}
As mentioned in Section \ref{ssec:previousworksub}, the main obstacle of the Dimension Spectrum Conjecture occurs when the dimension of $x$ is lower than the dimension of the line $(a,b)$. As mentioned in Section \ref{ssec:previousworksub}, in general, this issue is still formidable. However, in the Dimension Spectrum Conjecture, we are given the freedom to specifically construct the point $x$, allowing us overcome this obstacle.

The most natural way to construct a sequence $x$ with $\dim^{a,b}(x) = s$ is to start with a random sequence, and pad it with long strings of zeros. This simple construction, unfortunately, does not seem to work. 

We are able to overcome the obstacle by padding the random sequence \textit{with the bits of $a$}, instead of with zeros. Thus, given an approximation $(x, ax+b)$ we trivially have a decent approximation of $a$ (formalized iin Lemma \ref{lem:computeAB}). This allows us, using Lemma \ref{lem:point}, to restrict our search for $(a,b)$ to a smaller set of candidate lines.
\subsection{Proof for low-dimensional lines}

Fix a slope-intercept pair $(a,b)$, and let $d = \dim(a,b)$. Let $s \in (0,d)$. Let $y \in \R$ be random relative to $(a, b)$. Thus, for every $r \in \N$,
\begin{center}
$K^{a,b}_r(y) \geq r - O(\log r)$.
\end{center}
Define the sequence of natural numbers $\{h_j\}_{j \in \N}$ inductively as follows. Define $h_0 = 1$. For every $j > 0$, let 
		\begin{equation*}
		h_j = \min\left\{h \geq 2^{h_{j-1}}: K_h(a, b) \leq \left(d + \frac{1}{j}\right)h\right\}.
		\end{equation*}
Note that $h_j$ always exists. For every $r \in \N$, let
\begin{align*}
x[r] = \begin{cases}
a[ r - \lfloor sh_j\rfloor] &\text{ if } r \in (\lfloor sh_j\rfloor, h_j] \text{ for some } j \in \N \\
y[r] &\text{ otherwise}
\end{cases}
\end{align*}
where $x[r]$ is the $r$th bit of $x$. Define $x \in \R$ to be the real number with this binary expansion.

One of the most important aspects of our construction is that we encode (a subset of) the information of $a$ into our point $x$. This is formalized in the following lemma.
\begin{lem}\label{lem:computeAB}
For every $j\in\N$, and every $r$ such that $sh_j < r \leq h_j$,
\begin{center}
$K_{r- sh_j, r}(a,b \, | \, x, ax + b) \leq O(\log h_j)$.
\end{center}
\end{lem}
\begin{proof}
By definition, the last $r-sh_j$ bits of $x$ are equal to the first $r-sh_j$ bits of $a$. That is,
\begin{align*}
x[sh_j] \, x[sh_j + 1]\, \ldots x[r] &= a[0] \, a[1]\, \ldots a[r - sh_j]\\
&= a \uhr (r-sh_j).
\end{align*}
Therefore, since additional information cannot increase Kolmogorov complexity,
\begin{align*}
K_{r- sh_j, r}(a \, | \, x, ax + b) &\leq K_{r- sh_j, r}(a \, | \, x)\\
&\leq O(\log h_j).
\end{align*}
Note that, given $2^{-(r-sh_j)}$-approximations of  $a$, $x$, and $ax + b$, it is possible to compute an approximation of $b$. That is,
\begin{center}
$K_{r-sh_j}(b \, | \, a, x, ax+b) \leq O(\log h_j)$.
\end{center}
Therefore, by Lemma \ref{lem:unichain} and the two above inequalities,
\begin{align*}
K_{r- sh_j, r}(a,b \, | \, x, ax + b) &= K_{r- sh_j, r}(a \, | \, x, ax+ b)  \\
&\;\;\;\;\;\;\;\;\;\;\;\;+ K_{r- sh_j, r}(b \, | \, a, x, ax+ b)  + O(\log r) \tag*{}\\
&\leq O(\log h_j) + K_{r- sh_j, r}(b \, | \, a, x, ax+ b) + O(\log r)\\
&\leq O(\log h_j).
\end{align*}
\end{proof}

The other important property of our construction is that $(a,b)$ gives no information about $x$, beyond the information specifically encoded into $x$. 
\begin{lem}\label{lem:boundOnXAB}
For every $j\in\N$, the following hold.
\begin{enumerate}
\item $K^{a,b}_t(x) \geq t - O(\log h_j)$, for all $t \leq sh_j$.
\item $K^{a,b}_r(x) \geq sh_j + r - h_j - O(\log h_j)$, for all $h_j \leq r \leq sh_{j+1}$.
\end{enumerate}
\end{lem}
\begin{proof}
We first prove item (1). Let $t \leq sh_j$. Then, by our construction of $x$, and choice of $y$,
\begin{align*}
K^{a,b}_t(x) &\geq K^{a,b}_t(y) - h_{j-1} - O(\log t)\\
&\geq t - O(\log t) - \log h_j - O(\log t)\\
&\geq t - O(\log h_j).
\end{align*}

For item (2), let $h_j \leq r \leq sh_{j+1}$. Then, by item (1), Lemma \ref{lem:unichain} and our construction of $x$, 
\begin{align*}
K^{a,b}_r(x) &= K^{a,b}_{h_j}(x) + K^{a,b}_{r,h_j}(x) - O(\log r) \tag*{[Lemma \ref{lem:unichain}]}\\
&\geq sh_j + K^{a,b}_{r,h_j}(x) - O(\log r) \tag*{[Item (1)]}\\
&\geq sh_j + K^{a,b}_{r,h_j}(y) - O(\log r)\\
&\geq sh_j + r-h_j - O(\log r),
\end{align*}
and the proof is complete. 

\end{proof}

We now prove bounds on the complexity of our constructed point. We break the proof into two parts. 

In the first, we give lower bounds on $K_r(x, ax+b)$ at precisions $sh_j < r \leq h_j$. Intuitively, the proof proceeds as follows. Since $r > sh_j$, given $(x,ax+b)$ to precision $r$ immediately gives a $2^{-r +sh_j}$ approximation of $(a,b)$. Thus, we only have to search for candidate lines $(u,v)$ which satisfy $t = \|(a,b) - (u,v)\| < 2^{-r + sh_j}$. Then, because of the lower bound on $t$, the complexity $K_{r -t}(x)$ is maximal. In other words, we are essentially in the case that the complexity of $x$ is high. Thus, we are able to use the method described in Section \ref{ssec:previousworksub}. We now formalize this intuition.
\begin{lem}\label{lem:boundUpToH}
For every $\gamma > 0$ and all sufficiently large $j\in \N$,
\begin{center}
$K_r(x, ax + b) \geq (s + d)r - \gamma r$,
\end{center}
for every $r \in (sh_j, h_j]$.
\end{lem}
\begin{proof}
Let $\eta \in \Q$ such that 
\begin{center}
$d - \gamma/4 < \eta < d - \gamma^2$.
\end{center}
Let $\ve \in \Q$ such that
\begin{center}
$\ve  < \gamma (d - \eta) / 16$.
\end{center}
Note that
\begin{center}
    $\frac{4\ve}{1 - \eta} \leq \frac{\gamma}{4}$
\end{center}
We also note that, since $\eta$ and $\ve$ are constant,
\begin{center}
    $K(\eta, \ve) = O(1)$.
\end{center}
Let $D = D(r, (a,b), \eta)$ be the oracle of Lemma \ref{lem:oracles} and let $\delta = 1-\eta$.

Let $(u,v)$ be a line such that $t := \|(a,b) - (u,v)\| \geq r - sh_j$, and $ux + v = ax + b$. Note that $r - t \leq sh_j$. Then, by Lemma \ref{lem:lines}, Lemma \ref{lem:oracles}  and Lemma \ref{lem:boundOnXAB}(1), 

\begin{align}
K^D_r(u,v) &\geq K^D_t(a,b) + K^D_{r-t,r}(x \, | \, a,b) - O(\log r)\tag*{[Lemma \ref{lem:lines}]}\\
&\geq K^D_t(a,b) + K_{r-t,r}(x \, | \, a,b) - O(\log r)\tag*{[Lemma \ref{lem:oracles}]}\\
&\geq  K^D_t(a,b) +r - t - O(\log r) \tag*{[Lemma \ref{lem:boundOnXAB}(1)]}.
\end{align}
There are two cases by Lemma \ref{lem:oracles}. For the first, assume that $K^D_t(a,b) = \eta r$. Then
\begin{align}
&\geq \eta r  + r - t - O(\log r) \tag*{[Definition of $\dim$]}\\
&= (\eta -\ve)r + r - t \tag*{[$r$ is large]}\\
&\geq (\eta -\ve)r + (1-\eta)(r - t) \tag*{}
\end{align}
For the second, assume that $K^D_t(a,b) = K_t(a,b)$. Then
\begin{align}
K^D_r(u,v) &\geq K_t(a,b) +r - t - O(\log r) \tag*{}\\
&\geq dt - o(t)  +r - t - O(\log r) \tag*{[Definition of $\dim$]}\\
&= \eta r + (1-\eta)r - t(1 - d) - \ve r \tag*{[$r$ is large]}\\
&\geq \eta r -\ve r + (1-\eta)(r - t) \tag*{[$d > \eta$]}\\
&\geq (\eta - \ve) r + (1 - \eta)(r - t). 
\end{align}

Therefore, in either case, we may apply Lemma \ref{lem:point},
\begin{align}
K_r(x, ax + b) &\geq K_r^D(a,b,x)- K_{r-sh_j,r}(a,b \, | , x, ax+b)\tag*{[Lemma \ref{lem:point}]}\\
&\;\;\;\;\;\;\;\;\;\;\;\;\;\;\;\;\;\;\;\;-\frac{4\ve}{1-\eta}r- K(\eta, \ve) - O(\log r) \tag*{}\\
&\geq K_r^D(a,b,x)- K_{r-sh_j,r}(a,b \, | \, x, ax+b)-\frac{\gamma}{4}r-\frac{\gamma}{8}r \tag*{}\\
&=K_r^D(a,b,x)- K_{r-sh_j,r}(a,b \, | \, x, ax+b)-\frac{3\gamma}{8}r \label{eq:lowerBoundupToH}.
\end{align}
By Lemma \ref{lem:boundOnXAB}(1), our construction of oracle $D$, and the symmetry of information,
\begin{align}
K_r^D(a,b,x) &= K_r^D(a,b) + K_r^D(x \, | \, a,b) - O(\log r) \tag*{[Lemma \ref{lem:unichain}]}\\
&= K_r^D(a,b) + K_r(x \, | \, a,b) - O(\log r)\tag*{[Lemma \ref{lem:oracles}(ii)]} \\
&\geq \eta r + K_{r}(x \, | \, a,b) - O(\log r) \tag*{[Lemma \ref{lem:oracles}(i)]}\\
&\geq \eta r + sh_j - O(\log r) \tag*{[Lemma \ref{lem:boundOnXAB}(1)]}\\
&\geq \eta r + sh_j - \frac{\gamma}{4}r \label{eq:boundABXupToH}.
\end{align}
Finally, by Lemma \ref{lem:computeAB},
\begin{equation}
K_{r-sh_j,r}(a,b \mid x, ax+b) \leq \frac{\gamma}{8}r. \label{eq:boundABGivenXuptoH}
\end{equation}

Together, inequalities (\ref{eq:lowerBoundupToH}), (\ref{eq:boundABXupToH}) and (\ref{eq:boundABGivenXuptoH}) imply that
\begin{align*}
K_r(x, ax + b) &\geq K_r^D(a,b,x)- K_{r-sh_j,r}(a,b \, | , x, ax+b)-\frac{3\gamma}{8}r\\
&\geq \eta r + sh_j - \frac{\gamma}{4}r - \frac{\gamma}{8}r-\frac{3\gamma}{8}r\\
&\geq dr - \frac{\gamma}{4}r + sh_j - \frac{3\gamma}{4}r\\
&\geq dr + sh_j - \gamma r\\
&\geq (s + d)r - \gamma r,
\end{align*}
and the proof is complete.
\end{proof}

We now give lower bounds on the complexity of our point, $K_r(x, ax+b)$, when $h_j < r \leq sh_{j+1}$. Intuitively, the proof proeeds as follows. Using the previous lemma, we can, given a $2^{-h_j}$-approximation of $(x, ax +b)$, compute a $2^{-h_j}$-approximation of $(a,b)$. Thus, we only have to compute the last $r-h_j$ bits of $(a,b)$. Importantly, since $r > h_j$, the last $r - h_j$ bits of $x$ are maximal. Hence, we can simply lower the complexity of the last $r-h_j$ bits of $(a,b)$ to roughly $s(r - h_j)$. Thus, we are again, essentially, in the case where $\dim(x) \geq \dim(a,b)$ and the techniques of Section \ref{ssec:previousworksub} work. We now formalize this intuition.
\begin{lem}\label{lem:boundAfterH}
For every $\gamma > 0$ and all sufficiently large $j\in \N$,
\begin{center}
$K_r(x, ax + b) \geq (s + d)r - \gamma r$,
\end{center}
for every $r \in (h_j, sh_{j+1}]$.
\end{lem}
\begin{proof}
Recall that we are assuming that $s < d$.  Let $\hat{s} \in \Q\cap (0,s)$ be a dyadic rational such that 
\begin{center}
    $\gamma/8 < s - \hat{s} < \gamma / 4$.
\end{center}
Let $\hat{d} \in \Q \cap (0,\dim(a,b))$ be a dyadic rational such that 
\begin{center}
    $\gamma / 8 < \dim(a,b) - \hat{d} < \gamma / 4$. 
\end{center}
Define 
\begin{center}
$\alpha = \frac{s(r - h_j) + \dim(a,b)h_j}{r}$,
\end{center}
and $\eta \in \Q \cap (0, \alpha)$ by
\begin{center}
$\eta = \frac{\hat{s}(r - h_j) + \hat{d}h_j}{r}$.
\end{center}

Finally, let $\ve = \gamma^2/64$. Note that
\begin{align}
    \alpha - \eta &= \frac{s(r - h_{j}) + dh_{j} - \hat{s}(r - h_{j}) - \hat{d}h_{j}}{r}\tag*{}\\
    &= \frac{(s-\hat{s})(r - h_{j}) + (d - \hat{d})h_{j} }{r}\tag*{}\\
    &\leq \frac{\frac{\gamma}{4}(r - h_{j}) + \frac{\gamma}{4}h_{j} }{r}\tag*{}\\
    &= \frac{\gamma}{4}
\end{align}
Similarly,
\begin{align}
    \alpha - \eta &= \frac{s(r - h_{j}) + \dim(a,b)h_{j} - \hat{s}(r - h_{j}) - \hat{d}h_{j-1}}{r}\tag*{}\\
    &= \frac{(s-\hat{s})(r - h_{j}) + (\dim(a,b) - \hat{d})h_{j} }{r}\tag*{}\\
    &> \frac{\frac{\gamma}{8}(r - h_{j}) + \frac{\gamma}{8}h_{j} }{r}\tag*{}\\
    &= \frac{\gamma}{8}
\end{align}
In particular,
\begin{equation}
    \frac{4\ve}{\alpha - \eta} \leq \gamma / 4.
\end{equation}
We also note that
\begin{equation}
    K(\ve, \eta) \leq K(\gamma, \hat{s}, \hat{d}, r, h_{j}) \leq O(\log r),
\end{equation}
since $j$ was chosen to be sufficiently large and $\gamma$ is constant. 

Finally, let $D = D(r, (a,b), \eta)$ be the oracle of Lemma \ref{lem:oracles}. Note that we chose $D$ so that, roughly, $D$ lowers the complexity of the last $r-h_j$ bits of $(a,b)$ to $s(r - h_j)$. 

Let $(u,v)$ be a line such that $t := \|(a,b) - (u,v)\| \geq h_j$, and $ux + v = ax + b$. Then, by Lemmas \ref{lem:lines}, \ref{lem:oracles} and \ref{lem:boundOnXAB},
\begin{align}
K^D_r(u,v) &\geq K^D_t(a,b) + K^D_{r-t,r}(x \, | \, a,b) - O(\log r) \tag*{[Lemma \ref{lem:lines}]}\\
&\geq K^D_t(a,b) + K_{r-t,r}(x \, | \, a,b) - O(\log r) \tag*{[Lemma \ref{lem:oracles}]}\\
&\geq K^D_t(a,b) + s(r - t) - O(\log r)\tag*{[Lemma \ref{lem:boundOnXAB}(1)]}.
\end{align}
There are two cases. In the first, $K^D_t(a,b) = \eta r$. Then,
\begin{align*}
    K^D_r(u,v) &\geq \eta r + s(r - t) - O(\log r)\\
    &\geq (\eta - \ve) r + s(r-t)\\
    &\geq (\eta - \ve) r + (\alpha - \eta)(r-t).
\end{align*}
In the other case, $K^D_t(a,b) = K_t(a,b)$. Then,
\begin{align}
K^D_r(u,v) &\geq K_t(a,b) + s(r-t) - O(\log r)\tag*{}\\
&\geq dt - o(t) + s(r-t) - O(\log r) \tag*{[Definition of $\dim$]}\\
&= dh_j + d(t - h_j) + s(r - t) -o(r)\tag*{}\\
&= dh_j + d(t - h_j) + s(r - h_j) - s(t-h_j) - o(r)\tag*{}\\
&= \alpha r + (d - s)(t - h_j) - o(r)\tag*{}\\
&= \eta r + (\alpha - \eta)r + (d - s)(t - h_j) - o(r)\tag*{}\\
&\geq \eta r + (\alpha - \eta)(r - t) - o(r)\tag*{}\\
&\geq (\eta - \ve)r + (\alpha - \eta)(r - t)\tag*{}.
\end{align}

Therefore we may apply Lemma \ref{lem:point}, which yields
\begin{align}
K^D_r(a,b,x) &\leq K_r(x, ax + b) + K^D_{h_j,r}(a,b,x \, | \, x, ax + b) \tag*{[Lemma \ref{lem:point}]}\\
&\;\;\;\;\;\;\;\;\;\;\;+\frac{4\ve}{\alpha - \eta}r+K(\ve,\eta)+O(\log r)\tag*{}\\
&\leq K_r(x, ax + b) + K^D_{h_j,r}(a,b,x \, | \, x, ax + b) \tag*{}\\
&\;\;\;\;\;\;\;\;\;\;\;\;+ \frac{\gamma}{4}r+\frac{\gamma}{8}r \tag*{[Choice of $\eta, \ve$]}\\
&= K_r(x, ax + b) + K^D_{h_j,r}(a,b,x \, | \, x, ax + b)+\frac{3\gamma}{8}r \label{eq:boundABXgivenXAXB}.
\end{align}
By Lemma \ref{lem:boundOnXAB}, and our construction of oracle $D$,
\begin{align}
K_r^D(a,b,x) &= K_r^D(a,b) + K_r^D(x \, | \, a,b) - O(\log r) \tag*{[Lemma \ref{lem:unichain}]}\\
&= \eta r + K_r(x \, | \, a,b) - O(\log r) \tag*{[Lemma \ref{lem:oracles}]}\\
&\geq \eta r + sh_j + r - h_j - O(\log r) \tag*{[Lemma \ref{lem:boundOnXAB}(2)]} \\
&\geq \alpha r - \frac{\gamma}{4}r + sh_j + r - h_j - O(\log r) \tag*{}\\
&\geq s(r - h_j) + dh_j - \frac{\gamma}{4}r + sh_j + r - h_j - O(\log r)\tag*{}\\
&\geq (1+s)r - (1 - d)h_j - \frac{\gamma}{4}r \label{eq:boundABX2}.
\end{align}

By Lemmas \ref{lem:boundUpToH}, and \ref{lem:unichain}, and the fact that additional information cannot increase Kolmogorov complexity
\begin{align}
K_{h_j, r}(a, b, x \, | \, x, ax + b) &\leq K_{h_j,h_j}(a,b,x \, | \, x, ax + b) \tag*{}\\
&= K_{h_j}(a,b,x)- K_{h_j}(x, ax + b) \tag*{[Lemma \ref{lem:unichain}]}\\
&= K_{h_j}(a,b) + K_{h_j}(x\mid a,b)\tag*{}\\
&\;\;\;\;\;\;\;\;\;\;\;\;\;\;\;\;- K_{h_j}(x, ax + b) \tag*{[Lemma \ref{lem:unichain}]}\\
&= K_{h_j}(a,b) + sh_j - K_{h_j}(x, ax + b) \tag*{[Lemma \ref{lem:boundOnXAB}]}\\
&\leq K_{h_j}(a,b) +sh_j -(s+d)h_j + \frac{\gamma}{16}h_j \tag*{[Lemma \ref{lem:boundUpToH}]}\\
&\leq dh_j + h_j/j - dh_j + \frac{\gamma}{16}r\tag*{[Definition of $h_j$]}\\
&\leq \frac{\gamma}{8}r \label{eq:boundABXgivenX2}
\end{align}
Combining inequalities (\ref{eq:boundABXgivenXAXB}), (\ref{eq:boundABX2}) and (\ref{eq:boundABXgivenX2}) , we see that
\begin{align*}
K_r(x, ax + b) &\geq K_r^D(a,b,x) -\frac{\gamma}{8}r -\frac{3\gamma}{8}r \\
&\geq (1+s)r - (1-d)h_j -\frac{\gamma}{4}r -  \frac{\gamma}{4}r\\
&\geq (1+s)r - (1-d)h_j - \gamma r.
\end{align*}

Note that, since $d \leq 1$, and $h_j \leq r$, 
\begin{align*}
(1+s)r - h_j(1 - d) - (s + d)r &= r(1-d) - h_j(1-d)\\
&=(r - h_j)(1-d)\\
&\geq 0.
\end{align*}
Thus,  
\begin{align*}
K_r(x, ax+b) &\geq (1+s)r - h_j(1 - d)- \gamma r\\
&\geq (s + d)r - \gamma r,
\end{align*}
and the proof is complete for the case $s < \dim(a,b)$.

\end{proof}

We are now able to prove our main theorem.
\begin{proof}[Proof of Theorem \ref{thm:maintheoremLow}]
Let $(a,b) \in \R^2$ be a slope-intercept pair with 
\begin{center}
$d = \dim(a,b) \leq 1$.
\end{center}
Let $s \in [0, 1]$. If $s = 0$, then 
\begin{align*}
K_r(a, a^2 + b) &= K_r(a) + K_r(a^2 + b \, | \, a) + O(\log r)\\
&= K_r(a) + K_r(b \, | \, a) + O(\log r)\\
&= K_r(a,b) + O(\log r),
\end{align*}
and so the conclusion holds.

If $s = 1$, then by \cite{LutStu17}, for any point $x$ which is random relative to $(a,b)$,
\begin{center}
$\dim(x, ax + b) = 1 + d$,
\end{center}
and the claim follows.

If $s \geq d$, then Lutz and Stull \cite{LutStu17} showed that for any $x$ such that 
\begin{center}
    $\dim^{a,b}(x) = \dim(x) = s$,
\end{center}
we have $\dim(x, ax + b) = s + d$. 

Therefore, we may assume that $s \in (0, 1)$ and $s < d$. Let $x$ be the point constructed in this section. Let $\gamma > 0$. Let $j$ be large enough so that the conclusions of Lemmas \ref{lem:boundUpToH} and \ref{lem:boundAfterH} hold for these choices of $(a,b)$, $x$, $s$ and $\gamma$. Then, by Lemmas \ref{lem:boundUpToH} and \ref{lem:boundAfterH},
\begin{align*}
\dim(x, ax + b) &= \liminf_{r\rightarrow\infty} \frac{K_r(x, ax+ b)}{r}\\
&\geq \liminf_{r\rightarrow\infty} \frac{(s + d)r - \gamma r}{r}\\
&= s+d - \gamma.
\end{align*}
Since we chose $\gamma$ arbitrarily, we see that 
\begin{center}
$\dim(x,ax+b) \geq s + d$.
\end{center}
For the upper bound, let $j\in\N$ be sufficiently large. Then
\begin{align*}
K_{h_j}(x, ax + b) &\leq K_{h_j}(x, a, b)\\
&= K_{h_j}(a,b) + K_{h_j}(x \, | \, a,b)\\
&\leq dh_j + sh_j\\
&= (d+s)h_j.
\end{align*}
Therefore,
\begin{center}
$\dim(x, ax+b) \leq s + d$,
\end{center}
and the proof is complete.
\end{proof}

\section{High-Dimensional Lines}\label{sec:highdim}
In this section we show that the Dimension Spectrum Conjecture holds for lines of high dimension, i.e., for lines $L_{a,b}$ such that $\dim(a, b) > 1$. That is, we will prove the following theorem.
\begin{thm}\label{thm:maintheoremHigh}
Let $(a,b) \in \R^2$ be a slope-intercept pair with $\dim(a,b) > 1$.. Then for every $s \in [0, 1]$, there is a point $x\in \R$ such that
\begin{center}
$\dim(x, ax + b) = 1 + s$.
\end{center}
\end{thm}

\subsection{Overview of proof}
In this case, we again apply essential insight of the proof for low-dimensional lines, namely, encoding (a subset of) the information of $a$ into $x$. However, when $\dim(a,b) > 1$ constructing $x$ as before potentially causes a problem. Specifically, in this case, the previous construction might cause $\dim(x, ax +b)$ to become too large.

The overcome this, we rely on a non-constructive argument. More specifically, we begin as in the construction of $x$ in the low-dimensional case. However at stage $j$ of our construction, we do not add all $h_j - sh_j$ bits of $a$ to $x$. Instead we consider the $m = h_j - sh_j$ strings $\mathbf{x_0},\ldots, \mathbf{x_m}$, where
\begin{align*}
\mathbf{x_n}[i] = \begin{cases}
0 &\text{ if } 0 \leq i < m - n\\
\frac{1}{a}[i - (m-n)] &\text{ if } m-n\leq i\leq m \tag*{(*)}
\end{cases}
\end{align*}
and look at the extension of $x$ with the bits of $\mathbf{x_n}$.

Using a discrete, approximate, version of the mean value theorem, we are able to conclude that there is some extension $x^\prime = x\mathbf{x_n}$ such that
\begin{center}
    $\min\limits_{sh_j\leq r\leq h_j}\vert K_r(x^\prime, ax^\prime + b) - (1+s)r\vert $
\end{center}
is sufficiently small. We then carry on with the argument of the low-dimensional lines until $sh_{j+1}$.

\subsection{Proof for high-dimensional lines}
In order to prove Theorem \ref{thm:maintheoremHigh}, we will, given any slope-intercept pair $(a,b)$ and $s \in (0,1)$, construct a point $x \in [0, 1]$ such that $\dim(x, ax + b) = 1 + s$. 

Our construction is best phrased as constructing an infinite binary sequence $\mathbf{x}$, and then taking $x$ to be the unique real number whose binary expansion is $\mathbf{x}$. We now recall terminology needed in the construction. We will use bold variables to denote binary strings and (infinite) binary sequences. If $\mathbf{x}$ is a (finite) binary string and $\mathbf{y}$ is a binary string or sequence, we write $\mathbf{x} \prec \mathbf{y}$ if $\mathbf{x}$ is a prefix of $\mathbf{y}$.

Let $(a,b)$ be a slope intercept pair and let $d = \dim(a,b)$. Define the sequence of natural numbers $\{h_j\}_{j \in \N}$ inductively as follows. Define $h_0 = 2$. For every $j > 0$, let 
		\begin{equation*}
		h_j = \min\left\{h \geq 2^{h_{j-1}}: K_h(a, b) \leq \left(d + 2^{-j}\right)h\right\}.
		\end{equation*}
We define our sequence $\mathbf{x}$ inductively. Let $\mathbf{y}$ be a random, relative to $(a, b)$, binary sequence. That is, there is some constant $c$ such that
\begin{equation}\label{eq:lowerBoundOnY}
    K^{a,b}(y \uhr r) \geq r - c,
\end{equation}
for every $r \in \N$. We begin our inductive definition by setting $\mathbf{x}[0\ldots 2] = \mathbf{y}[0\ldots 2]$. Suppose we have defined $\mathbf{x}$ up to $h_{j-1}$. We then set
\begin{center}
    $\mathbf{x}[r] = \mathbf{y}[r]$, for all $h_{j-1} < r \leq sh_j$.
\end{center}

To specify the next $h_j - sh_j$ bits of $\mathbf{x}$, we use the following lemma, which we will prove in the next section.
\begin{lem}\label{lem:goodpointz}
For every sufficiently large $j$, there is a binary string $\mathbf{z}$ of length $h_j - sh_j$ such that
\begin{equation*}
    \min\limits_{sh_j < r \leq h_j} \left| K_r(x, ax+b) - (1+s)r \right| < \frac{r}{j},
\end{equation*}
where $x$ is any real such that $\mathbf{x}\mathbf{z} \prec x$. Moreover, $\mathbf{z}$ is of the form (*) of Section 5.1.
\end{lem}

For now, we assume the truth of this lemma. If the current $j$ is not sufficiently large, take $\mathbf{z}$ to be the string of all zeros. Otherwise, if $j$ is sufficiently large, we let $\mathbf{z}$ be such a binary string. We then set
\begin{center}
    $\mathbf{x}[r] = \mathbf{z}[r-sh_j]$, for all $sh_j < r \leq h_j$,
\end{center}
completing the inductive step. Finally, we let $x_{a,b,s}$ be the real number with binary expansion $\mathbf{x}$.

\begin{prop}\label{prop:highBoundX}
Let $x = x_{a,b,s}$ be the real we just constructed. Then for every $j$,
\begin{enumerate}
    \item $K^{a,b}_{sh_j}(x) \geq sh_j - O(\log h_j)$, and
    \item $K_r(x \mid a,b) \geq sh_j + r - h_j$, for every $h_j \leq r < sh_{j+1}$.
\end{enumerate}
\end{prop}
\begin{proof}
To see (1), by our construction, $\mathbf{x}[h_{j-1}\ldots sh_j] = \mathbf{y}[h_{j-1}\ldots sh_j]$. Thus, by Corollary \ref{cor:trunc},
\begin{align*}
    K^{a,b}_{sh_j}(x) &= K^{a,b}(\mathbf{x}) - O(\log h_j)\\
    &\geq K^{a,b}(\mathbf{y}) - h_{j-1} - O(\log h_j)\\
    &\geq sh_j - h_{j-1} - O(\log h_j)\\
    &\geq sh_j -O(\log h_j).
\end{align*}

For item (2), by our construction $\mathbf{x}[h_{j}\ldots r] = \mathbf{y}[h_{j}\ldots r]$. Therefore,
by Corollary \ref{cor:trunc} and our construction of $x$,
\begin{align*}
    K_{r}(x\mid a,b) &\geq K^{a,b}_r(x)\\
    &= K^{a,b}(\mathbf{x}) - O(\log h_j)\\
    &\geq K^{a,b}(\mathbf{y}[0\ldots sh_j]\mathbf{z}\mathbf{y}[h_j\ldots r]) - h_{j-1} - O(\log h_j)\\
    &\geq K^{a,b}(\mathbf{y}[0\ldots sh_j]\mathbf{y}[h_j\ldots r])\\
    &\;\;\;\;\;\;\;\;\;- h_{j-1} - O(\log h_j)\tag*{[$\mathbf{z}$ computable from $a$]}\\
    &\geq sh_j + r - h_j - O(\log h_j) \tag*{[Definition of $\mathbf{y}$]}
\end{align*}
\end{proof}

We now show, again assuming Lemma \ref{lem:goodpointz}, that $\dim(x, ax+b) = 1 + s$, where $x = x_{a,b,s}$ is the point we have just constructed. 

We begin by proving an upper bound on $\dim(x, ax+ b)$. Note that this essentially follows from our choice of $\mathbf{z}$.
\begin{prop}\label{prop:highdimupperbound}
Let $(a,b)$ be a slope intercept pair, $s \in (0,1)$ and $\gamma \in \Q$ be positive. Let $x = x_{a,b,s}$ be the point we have just constructed. Then 
\begin{center}
$\dim(x, ax + b) \leq (1 + s) + \gamma$.
\end{center}
\end{prop}
\begin{proof}
Let $j$ be sufficiently large. By our construction of $x$, 
\begin{equation}\label{eq:boundOnIntervalSH}
    \min\limits_{sh_j < r \leq h_j} \left| K_r(x, ax+b) - (1+s)r \right| < \frac{\gamma r}{4} 
\end{equation}

Therefore,
\begin{align*}
    \dim(x, ax+b) &= \liminf\limits_{r} \frac{K_r(x, ax+b)}{r}\\
    &\leq \liminf\limits_{j} \min\limits_{sh_j < r \leq h_j} \frac{K_r(x, ax+b)}{r}\\
    &\leq \liminf\limits_{j} \min\limits_{sh_j < r \leq h_j} \frac{(1+s)r + \gamma r/4}{r}\\
    &= \liminf\limits_{j} \min\limits_{sh_j < r \leq h_j} 1+s + \gamma/4\\
    &= 1 + s + \frac{\gamma}{4}.
\end{align*}
\end{proof}

We break the proof of the lower bound on $\dim(x, ax+b)$ into two parts. In the first, we give lower bounds on $K_r(x,ax+b)$ on the interval $r\in (sh_j, h_j]$. Note that this essentially follows from inequality (\ref{eq:boundOnIntervalSH}).
\begin{prop}\label{prop:highlowerbound}
Let $(a,b)$ be a slope intercept pair, $s \in (0,1)$, $\gamma \in \Q$ be positive and $j$ be sufficiently large. Let $x = x_{a,b,s}$ be the point we have just constructed. Then 
\begin{center}
$K_r(x, ax + b) \geq (1 + s - \gamma)r$
\end{center}
for all $sh_{j} < r \leq h_j$
\end{prop}
\begin{proof}
Let $j$ be sufficiently large and $sh_{j} < r \leq h_j$. Then, by (\ref{eq:boundOnIntervalSH}),
\begin{center}
    $K_r(x, ax+b) \geq (1+s)r - \gamma r$.
\end{center}

\end{proof}

We now give lower bounds on $K_r(x, ax+b)$ on the interval $r \in (h_{j-1}, sh_j]$. The proof of this lemma is very similar to the analogous lemma for low-dimensional lines (Lemma \ref{lem:boundAfterH}). Intuitively, the proof is as follows. Using the previous lemma, we can, given a $2^{-h_j}$-approximation of $(x, ax +b)$, compute a $2^{-h_j}$-approximation of $(a,b)$ \textit{with a small amount of extra bits}. Having done so, we have to compute the last $r-h_j$ bits of $(a,b)$. Importantly, since $r > h_j$, the last $r - h_j$ bits of $x$ are maximal. Thus, we can simply lower the complexity of the last $r-h_j$ bits of $(a,b)$ so that the complexity of these bits is roughly $s(r - h_j)$. Thus, we are again, morally, in the case where $\dim(x) \geq \dim(a,b)$ and the techniques of Section \ref{ssec:previousworksub} work. We now formalize this intuition.
\begin{lem}\label{lem:highlowerboundHSh}
Let $(a,b)$ be a slope intercept pair, $s \in (0,1)$, $\gamma \in \Q$ be positive and $j$ be sufficiently large. Let $x = x_{a,b,s}$ be the point we have just constructed. Then 
\begin{center}
$K_r(x, ax + b) \geq (1 + s - \gamma)r$
\end{center}
for all $h_{j-1} < r \leq sh_j$
\end{lem}
\begin{proof}
Intuitively, we will use the approximation of $(x, ax+b)$ at precision $h_{j-1}$ to compute $(a, b)$ at precision $h_{j-1}$. Then we will only search for candidate lines within $2^{-h_{j-1}}$ of $(a,b)$. Formally, the argument proceeds as follows. 

We first show that we can compute $(a,b)$ to within $2^{-h_{j-1}}$ with an approximation of $(x, ax+b)$, with few additional bits of information. By Lemma \ref{lem:unichain} and inequality (\ref{eq:boundOnIntervalSH})
\begin{align}
K_{h_{j-1}, r}(a, b, x \, | \, x, ax + b) &\leq K_{h_{j-1},h_{j-1}}(a,b,x \, | \, x, ax + b) + O(\log h_j)) \tag*{}\\
&= K_{h_{j-1}}(a,b,x)- K_{h_{j-1}}(x, ax + b) \tag*{[Lemma \ref{lem:unichain}]} \\
&\leq K_{h_{j-1}}(a,b,x) - (1+s)h_{j-1}  + \frac{\gamma}{4}h_{j-1}  \tag*{[(\ref{eq:boundOnIntervalSH})]}\\
&= K_{h_{j-1}}(a,b) + K_{h_{j-1}}(x\mid a,b) \tag*{} \\
&\;\;\;\;\;\;\;\;\;\;\;\;\;\;\;\;\;\;\;\;\;- (1+s)h_{j-1}  + \frac{\gamma}{4}h_{j-1}  \tag*{[Lemma \ref{lem:unichain}]}\\
&\leq dh_{j-1} + h_j2^{-j}  + K_{h_{j-1}}(x\mid a,b) \tag*{}\\
&\;\;\;\;\;\;\;\;\;\;\;\;\;\;\;\;\;\;- (1+s)h_{j-1}  + \frac{\gamma h_{j-1}}{4} \tag*{[Definition $h_j$]}\\
&\leq dh_{j-1} + h_j2^{-j} + sh_{j-1} \tag*{}\\
&\;\;\;\;\;\;\;\;\;\;\;\;\;\;\;\;\;\;- (1+s)h_{j-1}  + \frac{\gamma h_{j-1}}{4} \tag*{[Proposition \ref{prop:highBoundX}]}\\
&\leq dh_j +sh_{j-1} \tag*{} \\
&\;\;\;\;\;\;\;\;\;\;\;\;\;- (1+s)h_{j-1} + \frac{\gamma h_{j-1}}{2}\tag*{[$j$ large]}\\
&\leq dh_j  - h_j + \frac{\gamma h_{j-1}}{2} \label{eq:boundHighDim3}.
\end{align}
Thus, we can, given a $2^{-r}$ approximation of $(x, ax+b)$, compute a $2^{-h_{j-1}}$-approximation of $(a, b)$ with   
\begin{center}
    $(d - 1)h_j + \frac{\gamma h_{j-1}}{2}$
\end{center}
additional bits of information. Knowing $(a,b)$ to precision $h_{j-1}$ allows us to search for candidate lines within $2^{-h_{j-1}}$ of $(a,b)$, i.e., using Lemma \ref{lem:point} with $m = h_{j-1}$. 

Let $\hat{s} \in \Q\cap (0,s)$ be a dyadic rational such that 
\begin{center}
    $\gamma/8 < s - \hat{s} < \gamma / 4$.
\end{center}
Let $\hat{d} \in \Q \cap (0,\dim(a,b))$ be a dyadic rational such that 
\begin{center}
    $\gamma / 8 < \dim(a,b) - \hat{d} < \gamma / 4$. 
\end{center}
Define 
\begin{center}
$\alpha = \frac{s(r - h_{j-1}) + dh_{j-1}}{r}$.
\end{center}
Define 
\begin{center}
$\eta = \frac{\hat{s}(r - h_{j-1}) + \hat{d}h_{j-1}}{r}$.
\end{center}
Finally, let $\ve = \gamma^2/64$. Note that
\begin{align}
    \alpha - \eta &= \frac{s(r - h_{j-1}) + dh_{j-1} - \hat{s}(r - h_{j-1}) - \hat{d}h_{j-1}}{r}\tag*{}\\
    &= \frac{(s-\hat{s})(r - h_{j-1}) + (d - \hat{d})h_{j-1} }{r}\tag*{}\\
    &\leq \frac{\frac{\gamma}{4}(r - h_{j-1}) + \frac{\gamma}{4}h_{j-1} }{r}\tag*{}\\
    &= \frac{\gamma}{4}\label{eq:boundAlphaEta}
\end{align}
Similarly,
\begin{align}
    \alpha - \eta &= \frac{s(r - h_{j-1}) + dh_{j-1} - \hat{s}(r - h_{j-1}) - \hat{d}h_{j-1}}{r}\tag*{}\\
    &= \frac{(s-\hat{s})(r - h_{j-1}) + (d - \hat{d})h_{j-1} }{r}\tag*{}\\
    &> \frac{\frac{\gamma}{8}(r - h_{j-1}) + \frac{\gamma}{4}h_{j-1} }{r}\tag*{}\\
    &= \frac{\gamma}{8}
\end{align}
In particular,
\begin{equation}\label{eq:boundonEpsilonOverAlphaMinusEta}
    \frac{4\ve}{\alpha - \eta} \leq \gamma / 4.
\end{equation}
We also note that
\begin{equation}\label{eq:Ketaconstant}
    K(\ve, \eta) \leq K(\gamma, \hat{s}, \hat{d}, r, h_{j-1}) \leq O(\log r),
\end{equation}
since $j$ was chosen to be sufficiently large and $\gamma$ is constant.

Let $D = D(r, (a,b), \eta)$ be the oracle of Lemma \ref{lem:oracles}. We now show that the conditions of Lemma \ref{lem:point} are satisfied for these choices $a, b, \eta, \ve, r$ and $\delta = \alpha - \eta$, $m = h_{j-1}$ and $A = D$.

Let $(u,v)$ be a line such that $t := \|(a,b) - (u,v)\| \geq h_{j-1}$, and $ux + v = ax + b$. Then, by Lemmas \ref{lem:lines}, \ref{lem:oracles}, and Proposition \ref{prop:highBoundX},
\begin{align}
K^D_r(u,v) &\geq K^D_t(a,b) + K^D_{r-t,r}(x \, | \, a,b) - O(\log r) \tag*{[Lemma \ref{lem:lines}]}\\
&\geq K^D_t(a,b) + K_{r-t,r}(x \, | \, a,b) - O(\log r) \tag*{[Lemma \ref{lem:oracles}]}\\
&\geq K^D_t(a,b) + s(r-t) - O(\log r) \tag*{[Proposition \ref{prop:highBoundX}]}.
\end{align}
By Lemma \ref{lem:oracles}, there are two cases. In the first, $K^D_t(a,b) = \eta r$, and so
\begin{align*}
    K^D_r(u,v) &\geq K^D_t(a,b) + s(r-t) - O(\log r)\\
    &= \eta r + s(r-t) - O(\log r)\\
    &\geq (\eta - \ve)r + s(r-t) \tag*{[$j$ is large]}\\
    &\geq (\eta - \ve)r + \delta(r-t) \tag*{[$\gamma$ is small]}
\end{align*}
In the second case, $K^D_t(a,b) = K_t(a,b)$, and so
\begin{align}
K^D_r(u,v) &\geq K^D_t(a,b) + s(r-t) - O(\log r) \tag*{}\\
&\geq dt - o(t) + s(r - t) - O(\log r)\tag*{[Definition of $\dim$]}\\
&= dh_{j-1} + d(t - h_{j-1}) + s(r - t) -o(r)\tag*{}\\
&= dh_{j-1} + d(t - h_{j-1}) + s(r - h_{j-1}) - s(t-h_{j-1}) - o(r)\tag*{}\\
&= \alpha r +  d(t - h_{j-1}) -  s(t-h_{j-1}) - o(r)\tag*{[Definition of $\alpha$]}\\
&= \eta r + (\alpha - \eta)r + (d - s)(t - h_{j-1}) - o(r)\tag*{}\\
&\geq \eta r + (\alpha - \eta)r- o(r)\tag*{[$d > 1$, $t > h_{j-1}$]}\\
&\geq \eta r + (\alpha - \eta)(r - t) - o(r) \tag*{[$\alpha > \eta$]}\\
&\geq (\eta - \ve)r + \delta (r -t) \tag*{[$j$ is large]}\\
\end{align}

Therefore, in either case, we may apply Lemma \ref{lem:point}, relative to $D$ which yields
\begin{align}
K^D_r(a,b,x) &\leq K_r(x, ax + b) + K_{h_j,r}(a,b,x \, | \, x, ax + b) \tag*{}\\
&\;\;\;\;\;\;\;+\frac{4\ve}{\alpha - \eta}r+K(\ve,\eta)+O(\log r)\tag*{[Lemma \ref{lem:point}]}\\
&\leq K_r(x, ax + b) + dh_j  - h_j + \frac{\gamma h_{j-1}}{2}\tag*{}\\
&\;\;\;\;\;\;\;+\frac{4\ve}{\alpha - \eta}r+K(\ve,\eta)+O(\log r) \tag*{[(\ref{eq:boundHighDim3})]}\\
&\leq K_r(x, ax + b) + dh_j  - h_j + \frac{\gamma h_{j-1}}{2}\tag*{}\\
&\;\;\;\;\;\;\;+\frac{4\ve}{\alpha - \eta}r+O(\log r) \tag*{[(\ref{eq:Ketaconstant})]}\\
&\leq K_r(x, ax + b) + dh_j  - h_j + \frac{\gamma h_{j-1}}{2}\tag*{}\\
&\;\;\;\;\;\;\;+\frac{\gamma r}{4}+O(\log r) \tag*{[(\ref{eq:boundonEpsilonOverAlphaMinusEta})]}\\
&\leq K_r(x, ax + b) + dh_j  - h_j +\frac{3\gamma r}{4}+O(\log r) \label{eq:boundHighDim1}
\end{align}
By Lemma \ref{lem:boundOnXAB}, and our construction of oracle $D$,
\begin{align}
K_r^D(a,b,x) &= K_r^D(a,b) + K_r^D(x \, | \, a,b) - O(\log r) \tag*{[Lemma \ref{lem:unichain}]}\\
&= \eta r + K_r(x \, | \, a,b) - O(\log r) \tag*{[Lemma \ref{lem:oracles}]}\\
&\geq \eta r + sh_j + r - h_j - O(\log r) \tag*{[Lemma \ref{lem:boundOnXAB}(2)]} \\
&\geq \alpha r - \frac{\gamma}{4}r + sh_j + r - h_j - O(\log r) \tag*{}\\
&\geq s(r - h_j) + dh_j - \frac{\gamma}{4}r + sh_j + r - h_j - O(\log r)\tag*{}\\
&\geq (1+s)r - (1 - d)h_j - \frac{\gamma}{4}r \label{eq:boundHighDim2}.
\end{align}

Rearranging (\ref{eq:boundHighDim1}) and combining this with (\ref{eq:boundHighDim2}), we see that
\begin{align*}
    K_r(x, ax+b) &\geq K^D_r(a,b,x) - dh_j  + h_j -\frac{3\gamma r}{4}-O(\log r) \tag*{[(\ref{eq:boundHighDim1})]}\\
    &\geq (1+s)r- (1 - d)h_j - \frac{\gamma}{4}r \\
    &\;\;\;\;\;\;\;\;\;\;\;\;\;\;\;\;- dh_j  + h_j -\frac{3\gamma r}{4}-O(\log r) \tag*{[(\ref{eq:boundHighDim2})]}\\
    &= (1+s)r -\gamma r -O(\log r) \tag*{}
\end{align*}

\end{proof}

We are now able to prove that the Dimension Spectrum Conjecture holds for high dimensional lines.
\begin{proof}[Proof of Theorem \ref{thm:maintheoremHigh}]
Let $(a,b) \in \R^2$ be a slope-intercept pair with 
\begin{center}
$d = \dim(a,b) > 1$.
\end{center}
Let $s \in [0, 1]$. In the case where $s = 0$, Turetsky showed (Theorem \ref{thm:Turetsky}) that $1 \in \spec(L_{a,b})$, i.e., there is a point $x$ such that $\dim(x,ax+b) = 1$. In the case where $s = 1$, Lutz and Stull \cite{LutStu17} showed than any point $x$ which is random relative to $(a,b)$ satisfies
\begin{center}
    $\dim(x, ax+b) = 2$.
\end{center}

Therefore, we may assume that $s \in (0, 1)$. Let $x = x_{a,b,s}$ be the point constructed in this section. By Propositions \ref{prop:highdimupperbound} and \ref{prop:highlowerbound} and Lemma \ref{lem:highlowerboundHSh}, for every $\gamma$, 
\begin{center}
    $\vert \dim(x, ax + b) - (1 + s)\vert < \gamma$.
\end{center}
Thus, by the definition of effective dimension, 
\begin{center}
    $\dim(x, ax + b) = 1 + s$,
\end{center}
and the proof is complete.
\end{proof}

\subsection{Proof of Lemma \ref{lem:goodpointz}}
To complete the proof of the main theorem of this section, we now prove Lemma \ref{lem:goodpointz}. Recall that this states that, for every $j$, after setting $\mathbf{x}[h_{j-1}\ldots sh_j] = \mathbf{y}[h_{j-1}\ldots sh_j] $,  the following holds.

\textbf{Lemma \ref{lem:goodpointz}.} For every sufficiently large $j$ there is a binary string $\mathbf{z}$ of length $h_j - sh_j$ such that
\begin{equation*}
    \min\limits_{sh_j < r \leq h_j} \left| K_r(x, ax+b) - (1+s)r \right| < \frac{r}{j},
\end{equation*}
where $x$ is any real such that $\mathbf{x}\mathbf{z} \prec x$. Moreover, $\mathbf{z}$ is of the form (*) of Section 5.1.

Let $m = h_j - sh_j$. For each $0 \leq n \leq m$, define the binary string $\mathbf{x_n}$ of length $m$ by
\begin{align*}
\mathbf{x_n}[i] = \begin{cases}
0 &\text{ if } 0 \leq i < m - n\\
\frac{1}{a}[i - (m-n)] &\text{ if } m-n\leq i\leq m
\end{cases}
\end{align*}
Thus, for example $\mathbf{x_0}$ is the binary string of $m$ zeros, while $\mathbf{x_m}$ is the binary string containing the $m$-bit prefix of $\frac{1}{a}$. 

Let $x$ be the real number such that $\mathbf{x}\mathbf{x_0} \prec x$, and whose binary expansion contains only zeros after $sh_j$. For each $1 \leq n \leq m$, let $x_n$ be the real number defined by
\begin{center}
    $x_n = x +2^{-h_j + n}/a$.
\end{center}
Therefore, for every $n$,
\begin{center}
    $(x_n, ax_n + b) = (x_n, ax + b + 2^{-h_j + n})$.
\end{center}

Since the binary expansion of $x$ satisfies $x[r] = 0$ for all $r \geq sh_j$, we have, for every $n$,
\begin{equation}\label{eq:xxnprecxn}
    \mathbf{x}\mathbf{x_n} \prec x_n
\end{equation}

In other words, the binary expansion of $x_n$ up to index $h_j$ is just the concatenation of $\mathbf{x}$ and $\mathbf{x_n}$.

We now collect a few facts about our points $x_n$. 
\begin{lem}\label{lem:propertyofxn}
For every $n, r$ such that $0\leq n \leq m$ and $sh_j \leq r \leq h_j$ the following hold.
\begin{enumerate}
\item $K_{n,h_j}(a \, | \, x_n) \leq O(\log h_j)$.
\item For every $n$ and $n^\prime > n$,
\begin{center}
$\vert K_r(x_{n^\prime}, ax_{n^\prime} + b) -K_r(x_n, ax_n + b) \vert < n^\prime - n + \log(r)$.
\end{center}
\item $K_{r - sh_j, r}(a, b \mid x_m, ax_m + b) \leq O(\log r)$.
\end{enumerate}
Note that the constants implied by the big oh notation depend only on $a$.
\end{lem}
\begin{proof}
From the definition of $\mathbf{x_n}$,
\begin{center}
$\mathbf{x_n}[(m-n)\ldots m] = \frac{1}{a}[0\ldots n]$.
\end{center}
Since the map $a \mapsto \frac{1}{a}$ is bi-Lipschitz on an interval, there some constant $c$ depending only on $a$, such that we can compute a $2^{-n + c}$-approximation of $a$ given the first $n$ bits of $1/a$. Thus, by Corollary \ref{cor:trunccond},
\begin{align*}
    K_{n, h_j}(a \mid x_n) &= K(a \uhr n \mid x_n \uhr h_j) + O(\log h_j)\tag*{[Corollary \ref{cor:trunccond}]}\\
    &=  K(a \uhr n \mid \mathbf{x}\mathbf{x_n} )+ O(\log h_j) \tag*{[(\ref{eq:xxnprecxn})]}\\
    &\leq  K(a \uhr n \mid \mathbf{x_n})+ O(\log h_j) \tag*{}\\
     &\leq  O(\log h_j).
\end{align*}

For item (2), let $n^\prime > n$. We first assume that $r < h_j - n^\prime$. Then 
\begin{align*}
\vert x_{n^\prime} - x_n \vert &= x + 2^{-h_j + n^\prime}/a - (x + 2^{-h_j + n}/a )\\
&= 2^{-h_j}(2^{n^\prime} - 2^{n})/a )\\
&= 2^{-h_j + n^\prime}(1 - 2^{n - n^\prime})/a \\
&\leq 2^{-h_j + n^\prime}/2a\\
&\leq O(2^{-h_j + n^\prime}).
\end{align*}
Thus,
\begin{align*}
    \|(x_{n^\prime}, ax_{n^\prime} + b) -(x_n, ax_n + b)\| &= \|(x_{n^\prime}, ax + b + 2^{-h_j + n^\prime}) \\
    &\;\;\;\;\;\;\;\;\;\;\;\;-(x_n, ax + b + 2^{-h_j + n^\prime})\|\\
    &\leq O(2^{-h_j + n^\prime}).
\end{align*}
Therefore, under our assumption that $r<h_j - n^\prime$,
\begin{center}
$K_r(x_{n^\prime}, ax_{n^\prime} + b) = K_r(x_n, ax_n + b) - O(1)$,
\end{center}
and the claim follows. Now assume that $r \geq h_j - n^\prime$. By Corollary \ref{cor:trunc}, 
\begin{align*}
    K_r(x_{n^\prime}, ax_{n^\prime} + b) &= K_r(x_{n^\prime}, ax + b + 2^{-h_j + n^\prime})\\
    &= K_r(\mathbf{x}\mathbf{x_{n^\prime}}, ax + b + 2^{-h_j + n^\prime}).
\end{align*}
Similarly,
\begin{align*}
    K_r(x_{n}, ax_{n} + b) &= K_r(x_{n}, ax + b + 2^{-h_j + n})\\
    &= K_r(\mathbf{x}\mathbf{x_n}, ax + b + 2^{-h_j + n}).
\end{align*}
By our definition of $\mathbf{x_n}$ and $\mathbf{x_{n^\prime}}$, $\mathbf{x_{n^\prime}}$ contains all the bits of $\frac{1}{a}$ that $\mathbf{x_{n}}$ has. Thus,
\begin{center}
    $K(\mathbf{x_n} \uhr r \mid \mathbf{x_{n^\prime}} \uhr r) = O(\log n^\prime)$.
\end{center}
Given the first $r$ bits of $\mathbf{x_n}$, we can compute the first $r$ bits of $\mathbf{x_{n^\prime}}$ by simply giving the $n^\prime - n$ bits of $\frac{1}{a}$ that $\mathbf{x_{n^\prime}}$ contains but $\mathbf{x_{n}}$ does not. Thus,
\begin{center}
    $K(\mathbf{x_{n^\prime}} \uhr r \mid \mathbf{x_n} \uhr r) \leq n^\prime - n + O(\log n^\prime)$.
\end{center}
Similarly, since $ax_n + b$ and $ax_{n^\prime} + b$ are translates of each other, to compute one given the other, we just need to give the logarithm of their distance. Therefore,
\begin{align*}
    K_r(ax + b + 2^{-h_j + n^\prime} \, | \, ax + b + 2^{-h_j + n}) &\leq O(\log h_j + \log n^\prime)\\
    K_r(ax + b + 2^{-h_j + n} \, | \, ax + b + 2^{-h_j + n^\prime}) &\leq O(\log h_j + \log n^\prime)\\
\end{align*}
Taken together, these four inequalities show that
\begin{align*}
    K_r(x_n, ax_n + b \mid x_{n^\prime}, ax_{n^\prime} + b) &= K_r(x_n\mid  x_{n^\prime}, ax_{n^\prime}+ b) \\
    &\;\;\;\;\;\;\;\;\;\;\;\; + K_r(ax_n + b \mid x_{n^\prime}, ax_{n^\prime} + b)\\
    &\leq O(\log n^\prime) + O(\log h_j + \log n^\prime).
\end{align*}
Similarly, 
\begin{align*}
    K_r(x_{n^\prime}, ax_{n^\prime} + b \mid x_{n}, ax_{n} + b) &= K_r(x_{n^\prime}\mid  x_{n}, ax_{n} + b)\\
    &\;\;\;\;\;\;\;\;\;\;\;+ K_r(ax_{n^\prime} + b \mid x_{n}, ax_{n} + b)\\
    &\leq n^\prime - n + O(\log n^\prime) + O(\log h_j + \log n^\prime),
\end{align*}
and the proof of the second item is complete.

For item (3), note that the first $r$ bits of $x_m$, 
\begin{center}
    $x_m \uhr r = \mathbf{x}\mathbf{x_m}[0\ldots (r - sh_j)]$.
\end{center}
By definition, $\mathbf{x_m}[0\ldots (r - sh_j)] = \frac{1}{a}[0\ldots (r-sh_j)]$. Thus, given a $2^{-r}$ approximation of $x_m$, we can compute a $2^{-r + sh_j + c}$-approximation of $a$. Finally, given  a $2^{-r + sh_j + c}$-approximation of $x_m$, $a$ and $ax_m + b$, we can compute a $2^{-r + sh_j + c}$-approximation of $b$.
\end{proof}

\begin{proof}[Poof of Lemma \ref{lem:goodpointz}]
The proof of this lemma is essentially (a discrete, approximate, version of) the mean value theorem. In particular, we will show that $x_0$ satisfies 
\begin{center}
    $K_r(x_0, ax_0 + b) \leq (1 + s)r$,
\end{center}
for every $r \in [sh_j, h_j]$, and  $x_m$ satisfies
\begin{center}
    $K_r(x_m, ax_m+ b) \geq (1 + s)r$,
\end{center}
for every $r\in [sh_j, h_j]$. By Lemma \ref{lem:propertyofxn}, the map $n \mapsto K_r(x_n, ax_n +b)$ is ``continuous", and therefore there must be a point $x_n$ which satisfies 
\begin{center}
    $ \min\limits_{sh_j < r \leq h_j} \left| K_r(x, ax+b) - (1+s)r \right| < \frac{r}{j}$.
\end{center}

We now formalize this intuition. For each $n$, define
\begin{center}
$M_n = \min\{ \frac{K_r(x_n, ax_n + b)}{r} \, | \, sh_j \leq r \leq h_j\}$.
\end{center}
To prove our theorem, it suffices to show that there is an $n$ such that
\begin{equation*}\label{eq:findN}
1 + s - \frac{1}{j} \leq M_n \leq 1+s + \frac{1}{j}.
\end{equation*}

To begin, note, by our construction of $h_j$, $j < \log h_j$. We also note that $K_{h_j}(x_0) \leq sh_j$, since the bits of $x_0$ after $sh_j$ are all zero. Therefore,
\begin{align*}
    K_{h_j}(x_0, ax_0 + b) &= K_{h_j}(x_0) + K_{h_j}(ax_0 + b \mid x_0) + O(\log h_j) \tag*{[Lemma \ref{lem:unichain}]}\\
    &\leq sh_j + K_{h_j}(ax_0 + b \mid x_0) + O(\log h_j)\\
    &\leq sh_j + K_{h_j}(ax_0 + b) + O(\log h_j)\\
    &\leq sh_j + h_j + O(\log h_j)\\
\end{align*}
where the constant of the $O(\log h_j)$ term is independent of $x_0$. Thus, $M_0 < 1 + s + \frac{1}{j}$, for all  $j$ such that $\frac{h_j}{j} > O(\log h_j)$.

We now show that
\begin{center}
$K_r(x_m, ax_m + b) \geq (1+s)r - \frac{r}{j}$,
\end{center}
for all sufficiently large $j$, and for every $sh_j \leq r \leq h_j$. 

Let $sh_j \leq r \leq h_j$, $\eta = 1 - \frac{1}{2j}$ and let $D = D(r, (a,b), \eta$. Let $\epsilon = \frac{1}{32j^2}$ and $\delta = 1-\eta$. We now show that the conditions of Lemma \ref{lem:point} are satisfied.

Let $(u,v)$ be a line such that $t := \|(a,b) - (u,v)\| \geq r - sh_j$, and $ux_m + v = ax_m + b$. Note that $r - t \leq sh_j$. Then, by Lemma \ref{lem:lines}, Lemma \ref{lem:oracles}  and Lemma \ref{lem:boundOnXAB}(1), 
\begin{align}
K^D_r(u,v) &\geq K^D_t(a,b) + K^D_{r-t,r}(x_m \, | \, a,b) - O(\log r)\tag*{[Lemma \ref{lem:lines}]}\\
&\geq K^D_t(a,b) + K_{r-t,r}(x_m \, | \, a,b) - O(\log r)\tag*{[Lemma \ref{lem:oracles}]}\\
&= \min\{\eta r, K_t(a,b)\} + K_{r-t,r}(x_m \, | \, a,b) - O(\log r)\tag*{[Lemma \ref{lem:oracles}]}\\
&= \min\{\eta r, K_t(a,b)\} + r - t - O(\log r)\tag*{[construction of $x$]}.
\end{align}
There are two cases guaranteed by Lemma \ref{lem:oracles}. If $\eta r \leq K_t(a,b)$, then 
\begin{align*}
K^D_r(u,v) &\geq \eta r + r - t - O(\log r)\\
&\geq (\eta - \epsilon) r + r - t\tag*{[$j$ is large]}\\
&\geq (\eta - \epsilon) r + \delta(r - t) \tag*{[$\delta < 1$]},
\end{align*}
since $\ve r \geq O(\log r)$. Hence, we may apply Lemma \ref{lem:point} in this case. Otherwise,  $K^D_t(a,b)  = K_t(a,b)$. Then
\begin{align*}
K^D_r(u,v) &\geq K_t(a,b) + r - t - O(\log r)\\
&\geq (d - 2^{-j})t + r - t - O(\log r)\tag*{[Definition of $h_j$]}\\
&= \eta r + (1-\eta)r - t(1 + 2^{-j} - d) - O(\log r) \tag*{}\\
&\geq \eta r + (1-\eta)r - t(1 - \eta) - O(\log r) \tag*{}\\
&= \eta r + \delta( r - t) - O(\log r)\\
&\geq (\eta - \ve)r + \delta(r-t).
\end{align*}

Therefore we may apply Lemma \ref{lem:point}. 
\begin{align}
K_r(x_m, ax_m + b) &\geq K_r^D(a,b,x_m)- K_{r-sh_j,r}(a,b \, | , x_m, ax_m +b)\tag*{[Lemma \ref{lem:point}]}\\
&\;\;\;\;\;\;\;\;\;\;\;\;\;\;\;\;\;\;\;\;-\frac{4\ve}{1-\eta}r-K(\epsilon, \eta) - O(\log r) \tag*{}\\
&= K_r^D(a,b,x_m) - O(\log r) \tag*{[Lemma \ref{lem:propertyofxn}(3)]}\\
&\;\;\;\;\;\;\;\;\;\;\;\;\;\;\;\;\;\;-\frac{4\ve}{1-\eta}r- K(\epsilon, \eta) - O(\log r) \tag*{}\\
&= K_r^D(a,b,x_m) - O(\log r) -\frac{1/8j^2}{1/2j}r- K(\epsilon, \eta) - O(\log r) \tag*{}\\
&= K_r^D(a,b,x_m) - O(\log r) -\frac{r}{4j}- K(\epsilon, \eta) - O(\log r) \tag*{}\\
&= K_r^D(a,b,x_m) - O(\log r) -\frac{r}{4j}- K(j) - O(\log r) \tag*{}\\
&= K_r^D(a,b,x_m) -\frac{r}{4j} - O(\log r) \tag*{}\\
&= K_r^D(a,b,x_m) - \frac{r}{2j}\tag*{}.
\end{align}
By our construction of oracle $D$, and the symmetry of information,
\begin{align}
K_r^D(a,b,x_m) &= K_r^D(a,b) + K_r^D(x_m \, | \, a,b) - O(\log r) \tag*{[Lemma \ref{lem:unichain}]}\\
&= K_r^D(a,b) + K_r(x_m \, | \, a,b) - O(\log r)\tag*{[Lemma \ref{lem:oracles}(ii)]} \\
&\geq \eta r + K_{r}(x_m \, | \, a,b) - O(\log r) \tag*{[Lemma \ref{lem:oracles}(i)]}\\
&\geq \eta r + sh_j - O(\log r) \label{eq:boundABXupToHa}.
\end{align}
Therefore, we have
\begin{align*}
K_r(x_m, ax_m + b) &\geq \eta r + sh_j - O(\log r) - \frac{r}{2j}\\
&\geq (1 + s)r - \frac{r}{2j}  - \frac{r}{2j}\\
&> (1+s)r - \frac{r}{j}.
\end{align*}
Thus we have shown that 
\begin{center}
$K_r(x_m, ax_m + b) > (1+s)r - \frac{r}{j}$,
\end{center}
for every $sh_j \leq r \leq h_j$.

If either
\begin{itemize}
\item $M_m \leq (1 + s) + \frac{1}{j} $, or
\item $M_0 \geq (1 + s) - \frac{1}{j} $
\end{itemize}
our claim would follow. We will therefore assume that
\begin{center}
$M_0 < 1 + s - \frac{1}{j}$, and $1 + s + \frac{1}{j} < M_m$.
\end{center}

Define
\begin{align*}
L &= \{n \, | \, M_n < 1 + s - \frac{1}{j}\}\\
G &= \{n \, | \, M_n > 1 + s + \frac{1}{j}\}.
\end{align*}
By our assumption, $L$ and $G$ are non-empty. Suppose that $L$ and $G$ partition $\{0,\ldots, m\}$. Then there is a $n$ such that $n\in L$ and $n+1\in G$. However, by Lemma \ref{lem:propertyofxn},
\begin{center}
$\vert K_r(x_{n+1}, ax_{n+1} + b) -K_r(x_n, ax_n + b) \vert \leq 1 + O(\log r)$,
\end{center}
for every $r$. Let $r$ be a precision testifying to $x_n\in L$. Then
\begin{align*}
    (1 + s - \frac{1}{j})r &> K_r(x_n, ax_n + b)\\
    &> K_r(x_{n+1}, ax_{n+1} + b) - 1 - O(\log r).
\end{align*}
That is,
\begin{align*}
    \frac{K_r(x_{n+1}, ax_{n+1} + b)}{r} &< 1 + s - \frac{1}{j} + \frac{1}{r} + \frac{O(\log r)}{r}\\
    &< 1 + s + \frac{1}{j},
\end{align*}
which contradicts our assumption that $x_{n+1} \in G$ and the proof is complete.
\end{proof}

\section{Conclusion and Future Directions}\label{sec:conclusion}
The behavior of the effective dimension of points on a line is not only interesting from the algorithmic randomness viewpoint, but also because of its deep connections to geometric measure theory. There are many avenues for future research in this area.

The results of this paper show that, for any line $L_{a,b}$, the dimension spectrum $\spec(L_{a,b})$ contains a unit interval. However, this is not, in general, a tight bound. It would be very interesting to have a more thorough understanding of the ``low end" of the dimension spectrum. Stull \cite{Stull18} showed that the Hausdorff dimension of points $x$ such that
\begin{center}
$\dim(x, ax + b) \leq \alpha + \frac{\dim(a,b)}{2}$
\end{center}
is at most $\alpha$. Further investigation of the low-end of the spectrum is needed.

It seems plausible that, for certain lines, the dimension spectrum contains an interval of length greater than one. For example, are there lines in the plane such that $\spec(L)$ contains an interval of length strictly greater than $1$?

Another interesting direction is to study the dimension spectrum of particular classes of lines. One natural class is the lines $L_{a,b}$ whose slope and intercept are both in the Cantor set. By restricting the lines to the Cantor set, or, more generally, self-similar fractals, might give enough structure to prove tight bounds not possible in the general case.

Additionally, the focus has been on the effective (Hausdorff) dimension of points. Very little is known about the  effective \textit{strong} dimension of points on a line. The known techniques do not seem to apply to this question. New ideas are needed to understand the strong dimension spectrum of planar lines.

Finally, it would be interesting to broaden this direction by considering the dimension spectra of other geometric objects. For example, can anything be said about the dimension spectrum of a polynomial?

\bibliography{DSCPL}

\begin{thebibliography}{10}

\bibitem{DouLutMauTeut14}
Randall Dougherty, Jack Lutz, R~Daniel Mauldin, and Jason Teutsch.
\newblock Translating the {C}antor set by a random real.
\newblock {\em Transactions of the American Mathematical Society},
  366(6):3027--3041, 2014.

\bibitem{DowHir10}
Rod Downey and Denis Hirschfeldt.
\newblock {\em Algorithmic Randomness and Complexity}.
\newblock Springer-Verlag, 2010.

\bibitem{GuLutMayMos14}
Xiaoyang Gu, Jack~H Lutz, Elvira Mayordomo, and Philippe Moser.
\newblock Dimension spectra of random subfractals of self-similar fractals.
\newblock {\em Annals of Pure and Applied Logic}, 165(11):1707--1726, 2014.

\bibitem{LiVit08}
Ming Li and Paul~M.B. Vit\'{a}nyi.
\newblock {\em An Introduction to {K}olmogorov Complexity and Its
  Applications}.
\newblock Springer, third edition, 2008.

\bibitem{Lutz03a}
Jack~H. Lutz.
\newblock Dimension in complexity classes.
\newblock {\em {SIAM} J. Comput.}, 32(5):1236--1259, 2003.

\bibitem{Lutz03b}
Jack~H. Lutz.
\newblock The dimensions of individual strings and sequences.
\newblock {\em Inf. Comput.}, 187(1):49--79, 2003.

\bibitem{LutLut17}
Jack~H. Lutz and Neil Lutz.
\newblock Algorithmic information, plane {K}akeya sets, and conditional
  dimension.
\newblock {\em ACM Trans. Comput. Theory}, 10(2):Art. 7, 22, 2018.

\bibitem{LutLut20}
Jack~H Lutz and Neil Lutz.
\newblock Who asked us? how the theory of computing answers questions about
  analysis.
\newblock In {\em Complexity and Approximation}, pages 48--56. Springer, 2020.

\bibitem{LutMay08}
Jack~H. Lutz and Elvira Mayordomo.
\newblock Dimensions of points in self-similar fractals.
\newblock {\em SIAM J. Comput.}, 38(3):1080--1112, 2008.

\bibitem{Lutz17}
Neil Lutz.
\newblock Fractal intersections and products via algorithmic dimension.
\newblock In {\em 42nd International Symposium on Mathematical Foundations of
  Computer Science (MFCS 2017)}, 2017.

\bibitem{LutStu17}
Neil Lutz and D.~M. Stull.
\newblock Bounding the dimension of points on a line.
\newblock In {\em Theory and applications of models of computation}, volume
  10185 of {\em Lecture Notes in Comput. Sci.}, pages 425--439. Springer, Cham,
  2017.

\bibitem{LutStu17b}
Neil Lutz and D.~M. Stull.
\newblock Dimension spectra of lines.
\newblock In {\em Unveiling dynamics and complexity}, volume 10307 of {\em
  Lecture Notes in Comput. Sci.}, pages 304--314. Springer, Cham, 2017.

\bibitem{LutStu18}
Neil Lutz and D.~M. Stull.
\newblock Projection theorems using effective dimension.
\newblock In {\em 43rd International Symposium on Mathematical Foundations of
  Computer Science (MFCS 2018)}, 2018.

\bibitem{Mayordomo02}
Elvira Mayordomo.
\newblock A {K}olmogorov complexity characterization of constructive
  {H}ausdorff dimension.
\newblock {\em Inf. Process. Lett.}, 84(1):1--3, 2002.

\bibitem{Mayordomo08}
Elvira Mayordomo.
\newblock Effective fractal dimension in algorithmic information theory.
\newblock In S.~Barry Cooper, Benedikt L{\"o}we, and Andrea Sorbi, editors,
  {\em New Computational Paradigms: Changing Conceptions of What is
  Computable}, pages 259--285. Springer New York, 2008.

\bibitem{Nies09}
Andre Nies.
\newblock {\em Computability and Randomness}.
\newblock Oxford University Press, Inc., New York, NY, USA, 2009.

\bibitem{Stull18}
D.~M. Stull.
\newblock Results on the dimension spectra of planar lines.
\newblock In {\em 43rd {I}nternational {S}ymposium on {M}athematical
  {F}oundations of {C}omputer {S}cience}, volume 117 of {\em LIPIcs. Leibniz
  Int. Proc. Inform.}, pages Art. No. 79, 15. Schloss Dagstuhl. Leibniz-Zent.
  Inform., Wadern, 2018.

\bibitem{Turetsky11}
Daniel Turetsky.
\newblock Connectedness properties of dimension level sets.
\newblock {\em Theor. Comput. Sci.}, 412(29):3598--3603, 2011.

\end{thebibliography}

\end{document}